\documentclass[12pt]{article}
\usepackage[small]{titlesec}
\usepackage{amssymb,amsmath,amsthm,mathtools}
\usepackage{enumerate}
\usepackage{hyperref}
\usepackage{cite}
\usepackage[mathscr]{euscript}
\usepackage[letterpaper,hmargin=3.7cm,vmargin=3.5cm]{geometry}

\usepackage{setspace}
\usepackage{graphicx}
\newtheorem{theorem}{Theorem}[section]
\newtheorem*{thm}{Theorem}

\newtheorem{lemma}[theorem]{Lemma}
\newtheorem{corollary}[theorem]{Corollary}

\theoremstyle{definition}
\newtheorem{definition}[theorem]{Definition}

\theoremstyle{remark}
\newtheorem{remark}{Remark}

\newcommand{\ds}[1]{\displaystyle{#1}}
\newcommand{\abs}[1]{\left\lvert #1 \right\rvert}
\newcommand{\norm}[1]{\left\lVert #1 \right\rVert}
\newcommand{\I}{{\rm i}}
\newcommand{\inner}[2]{\left\langle#1,#2\right\rangle}

\newcommand{\cH}{\mathcal{H}}
\newcommand{\cB}{\mathcal{B}}

\newcommand{\R}{{\mathbb R}}
\newcommand{\C}{{\mathbb C}}
\newcommand{\N}{{\mathbb N}}
\newcommand{\Z}{{\mathbb Z}}
\newcommand{\cc}[1]{\overline{#1}}
\newcommand{\ournewclass}{\mathscr{S}(\mathcal{H})}

\newcommand{\mb}[1]{\boldsymbol{#1}}

\renewcommand\tilde{\widetilde}
\renewcommand\hat{\widehat}
\newcommand{\doi}[1]{\href{http://dx.doi.org/#1}{doi:#1}}

\DeclareMathOperator{\im}{im}
\DeclareMathOperator{\dom}{dom}
\DeclareMathOperator{\Ker}{ker}

\DeclareMathOperator{\ran}{ran}
\DeclareMathOperator{\spec}{spec}
\DeclareMathOperator{\Span}{span}

\DeclareMathOperator{\assoc}{assoc}

\begin{document}
\begin{titlepage}
\title%
{\vspace{-1cm}A class of $n$-entire Schr\"odinger operators 
\footnotetext{%
}
\footnotetext{%
}
\\[2mm]}
\author{
\textbf{Luis O. Silva}\thanks{Partially supported by CONACYT (M{\'e}xico)
	through grant CB-2008-01-99100}
\\
\small Departamento de F\'{i}sica Matem\'{a}tica\\[-1.6mm]
\small Instituto de Investigaciones en Matem\'{a}ticas Aplicadas y
	en Sistemas\\[-1.6mm]
\small Universidad Nacional Aut\'{o}noma de M\'{e}xico\\[-1.6mm]
\small C.P. 04510, M\'{e}xico D.F.\\[-1.6mm]
\small \texttt{silva@iimas.unam.mx}
\\[4mm]
\textbf{Julio H. Toloza}\thanks{Partially supported by CONICET (Argentina)
	through grant PIP 112-201101-00245}
\\
\small CONICET\\[-1.6mm]
\small Centro de Investigaci\'{o}n en Inform\'{a}tica para la
	Ingenier\'{i}a\\[-1.6mm]
\small Universidad Tecnol\'{o}gica Nacional --
	 Facultad Regional C\'{o}rdoba\\[-1.6mm]
\small Maestro L\'{o}pez esq.\ Cruz Roja Argentina\\[-1.6mm]
\small X5016ZAA C\'{o}rdoba, Argentina\\[-1.6mm]
\small \texttt{jtoloza@scdt.frc.utn.edu.ar}}
\date{}
\maketitle
\begin{center}
\begin{minipage}{5in}
  \centerline{{\bf Abstract}}
  \bigskip
  We study singular Schr\"odinger operators on a finite interval as
  selfadjoint extensions of a symmetric operator. We give sufficient
  conditions for the symmetric operator to be in the $n$-entire class,
  which was defined in our previous work \cite{IV}, for some $n$. As a
  consequence of this classification, we obtain a detailed spectral
  characterization for a wide class of radial Schr\"odinger
  operators. The results given here make use of de Branges Hilbert
  space techniques.
\end{minipage}
\end{center}
\bigskip
\thispagestyle{empty}
\end{titlepage}


\section{Introduction}

This paper deals with the spectral analysis of selfadjoint
operators arising from the differential expression
\[
 -\frac{d^2}{dx^2} + \frac{l(l+1)}{x^2} + q(x),
\quad x\in(0,1),\quad  l\ge-\frac12,
\]
along with separated selfadjoint boundary conditions.  These operators
describe the radial part of the Schr\"odinger operator for a particle
confined to a ball of finite radius, when the potential is spherically
symmetric. Here we show that this kind of differential operators can
be related naturally to a class of symmetric operators recently
introduced in the literature \cite{IV}, the so-called $n$-entire
operators. As such, we are able to give a spectral characterization of
these operators for a wide class of functions $q(x)$. Namely, we show
(for $l>-1/2$) that if $xq(x)$ lies in $L_p(0,1)$, $p>2$, then the
spectra of any two selfadjoint realizations satisfy conditions of
convergence and have certain regularity in their distribution (for
details see Corollary~\ref{cor:main}; this result holds also for
$l=-1/2$).  To our knowledge, this characterization is new for this
class of operators.

This work has been motivated by several relatively new
developments. Firstly, recent developments on one-dimensional
Schr\"odinger operators with singular potentials, in particular, the
perturbed Bessel operators
\cite{albeverio,carlson,carlson2,eckhardt0,eckhardt,eckhardt2,everitt,fulton1,
  fulton2,gesztesy,guillot,kostenko,kostenko2,kostenko2.5,kostenko3,lesch,serier}.
Secondly, the approach by Remling to the inverse spectral analysis of
regular Schr{\"o}dinger operators by means of de Branges space theory
\cite{remling}. Finally, our generalization \cite{IV} of Krein's
theory of entire operators \cite{krein1,krein2,krein3} (for a review
see \cite{gorbachuk}), which in its turn was inspired by Woracek's
work on $n$-associated functions of a de Branges space
\cite{langer-woracek,woracek2}. We note in passing that de Branges
space techniques have also been used in conection with the inverse
spectral analysis of differential operators in
\cite{eckhardt0,eckhardt}.

Among the classes of $n$-entire operators one finds the classes of entire
operators and entire operators in the generalized sense. The latter
classes were introduced by Krein to treat various problems of
analysis, particularly, spectral theory \cite{gorbachuk}. The
construction of $n$-entire operators was possible due to recent
results in de Branges space theory \cite{woracek2}.

Our approach for dealing with the singular Schr\"odinger operators
under consideration is based on a particular kind of perturbative
method in which the conclusion of Theorem~\ref{thm:spaces-equality} is
a crucial fact.  Due to the functional model given in \cite{IV}, one
has that to any potential $q(x)$ (in the class considered), and any
$l\ge -1/2$, there corresponds a de Branges space. Remarkably, as
Theorem~\ref{thm:spaces-equality} establishes, the set of functions in
the de Branges space for the free operator ($q(x)\equiv 0$) equals the
set of functions in the de Branges space for the perturbed operator.
Since we also show in Theorem~\ref{thm:free-operator-is-n-entire} that
the unperturbed radial Schr\"odinger operator is $n$-entire as long as
$n>\frac{l}{2}+\frac{3}{4}$, Theorem~\ref{thm:spaces-equality} allows
us to extend this assertion to the whole range of potential functions
under consideration (Theorem~\ref{thm:main}). This in its turn implies
the distinctive distributional properties of the spectra of the
associated selfadjoint realizations (Corollary~\ref{cor:main}).  We
point out that Theorem~\ref{thm:main} and Corollary~\ref{cor:main}
illustrate an application of the notion of $n$-entire operators to the
direct spectral analysis of Schr\"odinger operators.

The exposition is organized as follows. In the next section we give
the preparatory material, lay out the notation, introduce the relevant
notions, and recall some formulae needed later. In Section 3 we study
the free particle case, that is, $q(x)\equiv 0$. Finally, in Section
4, we discuss the main results of this work.

\section{Preliminaries}

In this section we give a brief recollection of the main theoretical
ingredients that will be used in this paper. Concretely, we touch upon
the distinct features of $n$-entire operators, and give a brief account
on the theory of de Branges Hilbert spaces.


\subsection{de Branges Hilbert spaces}
L. de Branges introduced a class of Hilbert space of entire functions which has
several distinguishing characteristics. These spaces can be defined in
various ways. For our choice of the definition we need two
ingredients. The first one is the Hardy space
\begin{equation*}
  H_2^+:=\{f(z)\text{ is holomorphic in }\C^+: \sup_{y>0}\int_\R\abs{f(x+\I y)}^2dx<\infty\},
\end{equation*}
where $\C^+=\{z\in\C:\im z>0\}$. The second one is an Hermite-Biehler
function, viz., an entire function $e(z)$ satisfying
$\abs{e(z)}>\abs{e(\cc{z})}$ for all $z\in\C^+$. Now, define
\begin{equation*}
  \cB(e):=\{f(z)\text{ entire}: f(z)/e(z), f^\#(z)/e(z)\in H_2^+\},
\end{equation*}
where $f^\#(z)=\cc{f(\cc{z})}$. Then we call the space
$\cB(e)$ endowed with the inner product
\begin{equation}
  \label{eq:dB-inner-product}
  \inner{g}{f}:=\int_\R\frac{\cc{g(x)}f(x)}{\abs{e(x)}^2}dx
\end{equation}
the de Branges Hilbert space generated by the Hermite-Biehler function
$e(z)$.
Note that the inner product is well defined since both $f(z)/e(z)$ and
$f^\#(z)/e(z)$ belong to $H_2^+$.
\begin{remark}
\label{rem:alternative-dB}
In view of \cite[Theorem 5.19]{rr}, $\cB(e)$ is the set of all entire
functions $f$ satisfying
\begin{equation}
  \label{eq:dB-integrability}
  \int_\R\abs{\frac{f(x)}{e(x)}}^2dx<\infty
\end{equation}
and
\begin{equation}
  \label{eq:dB-Cauchy}
  \frac{1}{2\pi\I}\int_\R h_i(x)\frac{dx}{x-z}=
  \begin{cases}
    h_i(z),& z\in\C^+,\\
    0, & z\in\C^-.
  \end{cases}\quad i=1,2\,,
\end{equation}
where $h_1(z)=f(z)/e(z)$ and $h_2(z)=f^\#(z)/e(z)$.
\end{remark}

Alternative characterizations of de Branges Hilbert spaces are
provided in \cite[Proposition 2.1]{remling} and also \cite[Chapter
2]{debranges}.  It is also possible to define de Branges spaces
without relying on a given Hermite-Biehler function \cite[Problem
50]{debranges}. However, for all de Branges spaces considered below,
there is always an explicit Hermite-Biehler function.

To any de Branges space there corresponds a space of associated functions
\cite[Definition 4.4]{kaltenback}. In this work we are interested in the
generalization of this notion given in \cite{langer-woracek}. For any
$\cB(e)$ and $n\in\Z^+=\N\cup\{0\}$, let
\begin{equation}
\label{eq:n-assoc-functions}
\assoc_n\cB(e) := \cB(e) + z\cB(e)+\cdots+z^n\cB(e).
\end{equation}
This is the so-called space of $n$-associated functions that was
introduced in the context of intermediate Weyl coefficients.

\subsection[n-entire operators]{$\mb{n}$-entire operators}

Given a separable Hilbert space $\cH$, let $\ournewclass$ be the class
of regular, closed symmetric operators on $\cH$, whose deficiency
indices are both equal to $1$; we recall that a closed operator $A$ is
regular if for every $z\in\C$ there exists $k_z>0$ such that
$\norm{(A-zI)\varphi}\ge k_z\norm{\varphi}$ for every
$\varphi\in\dom(A)$.  Well known properties of operators of this kind
are discussed in \cite{IV}. Relevant to this work is that for any
operator in $\ournewclass$, one can construct a de Branges space
\cite[Section~2.3]{IV}.

\begin{definition}
Given $n\in\Z^+$, $A\in\ournewclass$ is $n$-entire if and only if there exist
$n+1$ vectors $\mu_0,\ldots,\mu_n\in\cH$ such that
\begin{equation}
\label{eq:n-entire}
\cH = \ran(A-zI)\dotplus\Span\{\mu_0+z\mu_1+\cdots+z^n\mu_n\},
\end{equation}
for all $z\in\C$.
\end{definition}

This definition extends the notions of entire operator and
entire operator in the generalized sense that were introduced
by M. G. Krein as tools to study several problems
in classical analysis \cite{gorbachuk}.

The next result is shown in \cite[Propositions 2.14, 3.7 and 3.11]{IV}.

\begin{theorem}
\label{thm:n-entire}
Let $A\in\ournewclass$. Then the following statements are
equivalent:
\begin{enumerate}
\item $A$ is $n$-entire.

\item The space of $n$-associated functions of the de Branges space
  	related to $A$ (see \cite[Section~2.3]{IV}) contains a 
  	zero-free function.

\item There exist $n+1$ vectors $\eta_0,\ldots,\eta_n\in\cH$ such that
	\[
	\cH = \ran(A-zI)\dotplus\Span\{\eta_0+z\eta_1+\cdots+z^n\eta_n\},
	\]
	for all $z\in\C$ with the possible exception of a finite number of
	points.

\item Let $A_{\beta_1}$ and $A_{\beta_2}$, $\beta_1\ne \beta_2$,
	be 	canonical selfadjoint extensions of $A$. Set
	$\{x_j\}_{j\in\mathbb{N}}
	=\{x_j^+\}_{j\in\mathbb{N}}\cup\{x_j^-\}_{j\in\mathbb{N}}
	=\spec(A_{\beta_1})$, where $\{x_j^+\}_{j\in\mathbb{N}}$ and
  	$\{x_j^-\}_{j\in\mathbb{N}}$ are the sequences of positive,
  	respectively non-positive, elements of $\spec(A_{\beta_1})$,
  	arranged according to increasing modulus. Then the following
  	assertions hold true:
	\begin{enumerate}[(C1)]
	\item The limit
		$\displaystyle{\lim_{r\to\infty}\sum_{0<|x_j|\le r}
		\frac{1}{x_j}}$
		exists.
	\item $\displaystyle{\lim_{j\to\infty}\frac{j}{x_j^{+}}
		=- \lim_{j\to\infty}\frac{j}{x_j^{-}}<\infty}$.
	\item Denoting $\spec(A_\beta)=\{b_j\}_{j\in\mathbb{N}}$ for an arbitrary
		canonical selfadjoint extension of $A$,
		define\\[3mm]
		$\displaystyle
		h_\beta(z):=\left\{\begin{array}{ll}
			\displaystyle{\lim_{r\to\infty}\prod_{|b_j|\le r}
			\left(1-\frac{z}{b_j}\right)}
				& \mbox{ if }0\not\in\sigma(A_\beta),
			\\
			\displaystyle{z\lim_{r\to\infty}\prod_{0<|b_j|\le r}
			\left(1-\frac{z}{b_j}\right)}
				& \mbox{ otherwise. }
			   \end{array}\right.
		$\\[2mm]
		The series
		$\displaystyle{
		\sum_{x_j\ne 0}\abs{\frac{1}
		{x_j^{2n}h_{\beta_2}(x_j)h_{\beta_1}'(x_j)}}}$ is
		convergent.
\end{enumerate}
\end{enumerate}
\end{theorem}

\begin{remark}
  In \cite[Theorem~3.2]{woracek2}, necessary and sufficient conditions
  for the existence of a real zero-free entire function in the space
  of $n$-associated functions of any de Branges space are
  provided. Therefore, assertions  2 and 4 of the previous theorem are
  connected via Woracek's result.
\end{remark}

The last assertion of the preceding theorem, can be interpreted as a
spectral characterization of $n$-entire operators. We use this
characterization for the class of Schr{\"o}dinger operators discussed
in this work, cf. Theorem~\ref{thm:main}.

\section{The radial Schr\"odinger operator: Free particle case}

In this section we consider the spectra of the selfadjoint operators
(realized by separated boundary conditions),
associated with the differential expression
\begin{equation}
\label{eq:differential-expression-free}
\tau_l := -\frac{d^2}{dx^2} + \frac{l(l+1)}{x^2},
	\quad x\in(0,1),\quad  l\ge-\frac12,
\end{equation}
from the perspective of Theorem~\ref{thm:n-entire}. In particular,
we study conditions for the validity of (C1), (C2), and (C3).

We start by recalling some basic facts. The differential expression
$\tau_l$ is regular at $x=1$. At $x=0$ it is in the limit
point case for $l\ge 1/2$ and limit circle case for $l\in[-1/2,1/2)$. In
the latter case it is common to add the boundary condition
\begin{equation}
\label{eq:boundary-condition-at-0}
\lim_{x\to 0}x^l\left[(l+1)\varphi(x) - x\varphi'(x)\right]=0.
\end{equation}
A comprehensive investigation of possible boundary conditions (at
$x=0$) can be found in \cite{bulla}.

The differential expression $\tau_l$, supplemented with the boundary
condition \eqref{eq:boundary-condition-at-0} when required, gives rise
to a closed, regular, symmetric operator with deficiency indices both
equal to $1$.
We will denote this operator as $H_l$. The
associated one-parameter family of canonical selfadjoint extensions
$H_{l,\beta}$, with $\beta\in[0,\pi)$, is described by
boundary conditions at the right endpoint \cite{albeverio,carlson}.
Namely,
\[
\dom(H_{l,\beta})
	:= \left\{\begin{gathered}\varphi(x)\in L^2(0,1)\cap\text{AC}^2(0,1]:
		(\tau_l\varphi)(x)\in L^2(0,1),\\
		\varphi(1)\cos\beta = \varphi'(1)\sin\beta
		\end{gathered}\right\}
\]
plus the boundary condition \eqref{eq:boundary-condition-at-0}
for $l\in[-1/2,1/2)$. Of course, one defines
$(H_{l,\beta}\varphi)(x)=(\tau_l\varphi)(x)$ in its domain.

The equation
\[
- \varphi''(x) + \frac{l(l+1)}{x^2}\varphi(x) = z\varphi(x)
\]
has a solution in $L^2(0,1)$, namely,
\begin{equation}
\label{eq:fundamental-solution-free}
\xi_l(z,x)
	:= z^{-\frac{2l+1}{4}}\sqrt{\frac{\pi x}{2}} J_{l+\frac12}(\sqrt{z}x),
\end{equation}
where $J_{m}(z)$ is the Bessel function of order $m$.
In what follows all cut branches are
chosen along the negative real axis. By \cite[9.1.10]{abramowitz},
$\xi_l(z,x)$ turns out to be an entire function (for every value
of $x$). In fact,
\begin{equation}
\label{eq:fundamental-solution-entire}
\xi_l(z,x) = \sqrt{\pi}\left(\frac{x}{2}\right)^{l+1}
	\underbrace{\sum_{k=0}^\infty\frac{\left(-\frac14 z x^2\right)^k}
	{k!\Gamma(l+k+\frac32)}}_{\ds{=:g_l(z,x)}};
\end{equation}
note that it satisfies \eqref{eq:boundary-condition-at-0} for
all $z\in\C$. It is worth to mention that $\xi_l(z,x)$ is a non vanishing
element of $\Ker(H^*_l-zI)$ for all $z\in\C$.


\begin{theorem}
\label{thm:free-operator-is-n-entire}
Given $l\ge -\frac12$, the operator $H_l$ is $n$-entire, $n\in\Z^+$,
if and only if $n>\frac{l}{2}+\frac34$.
\end{theorem}
\begin{proof}
  In view of Theorem~\ref{thm:n-entire}, it suffices to find two
  selfadjoint realizations of \eqref{eq:differential-expression-free}
  whose spectra satisfy conditions (C1), (C2) and (C3).  According to
  \cite{kostenko},
\[
\spec(H_{l,\beta})
	= \{\text{zeros of }\xi_l(z,1)\cos\beta-\xi'_l(z,1)\sin\beta\},
\]
where the prime denotes derivative with respect to $x$. For $\beta=0$
this becomes
\[
\spec(H_{l,0})
	= \big\{\text{zeros of }J_{l+\frac12}(\sqrt{z})\big\}
	= \big\{(j_{l+\frac12,n})^2: n\in\N\big\};
\]
here we are using the notation of \cite[Chapter 9]{abramowitz}.
For $\beta\in(0,\pi)$, we can write
\[
\spec(H_{l,\beta})
	= \{\text{zeros of } \xi'_l(z,1)-\xi_l(z,1)\cot\beta\}.
\]
Now,
\[
\xi'_l(z,x)
	= \frac12\sqrt{\frac{\pi}{2}}z^{-\frac{2l+1}{4}} x^{-\frac12}
				 	J_{l+\frac12}(\sqrt{z}x)
			    + \sqrt{\frac{\pi}{2}}z^{-\frac{2l-1}{4}} x^{\frac12}
				 	J'_{l+\frac12}(\sqrt{z}x).
\]
By invoking the recurrence relation $J'_m(z)=-J_{m+1}(z)+mz^{-1}J_m(z)$,
we obtain
\[
\xi'_l(z,1) = \sqrt{\frac{\pi}{2}}(l+1)
					z^{-\frac{2l+1}{4}}J_{l+\frac12}(\sqrt{z})
			    - \sqrt{\frac{\pi}{2}}z^{-\frac{2l-1}{4}}
				 	J_{l+\frac32}(\sqrt{z}).
\]
Thus,
\begin{multline*}
\xi'_l(z,1)-\xi_l(z,1)\cot\beta
	= \sqrt{\frac{\pi}{2}}(l+1-\cot\beta)
					z^{-\frac{2l+1}{4}}J_{l+\frac12}(\sqrt{z})
					\\
			    - \sqrt{\frac{\pi}{2}}z^{-\frac{2l-1}{4}}
				 	J_{l+\frac32}(\sqrt{z}).
\end{multline*}
The zeros of this expression are easy to describe if
$\beta=\beta_l$, where $\beta_l$ satisfies $\cot\beta_l=l+1$.
In this case we can write
\[
\xi'_l(z,1)-\xi_l(z,1)\cot\beta_l = - \sqrt{\frac{\pi}{2}}z
		\underbrace{z^{-\frac{2l+3}{4}}J_{l+\frac32}(\sqrt{z})}_
		{\mathclap{\ds{\text{proportional to } g_{l+1}(z,1)}}},
\]
where $g_l(z,x)$ is defined as in \eqref{eq:fundamental-solution-entire} and
clearly $g_l(0,1)\ne 0$. Therefore,
\[
\spec(H_{l,\beta_l})
	= \big\{\text{zeros of } zg_{l+1}(z)\big\}
	= \big\{0\big\}\cup\big\{(j_{l+\frac32,n})^2: n\in\N\big\}.
\]

In what follows we will consider $\spec(H_{l,0})$ and $\spec(H_{l,\beta_l})$.
At this point it is convenient to recall the
asymptotic formulae \cite[9.2.1 and 9.5.12]{abramowitz},
\begin{gather}
J_{l+\frac32}(u)
	= \sqrt{\frac{2}{\pi u}}\left[\cos\left(u-\tfrac{\pi}{2}l-\pi\right)
					+ O\left(u^{-1}\right)\right],\quad u\to +\infty,
	\label{eq:asymptotics-for-bessel}
\\[4mm]
\sqrt{x_m} = j_{l+\frac12,m}
		   = \frac{\pi}{2}(2m+l) + O\left(m^{-1}\right),\quad m\to +\infty.
		     \label{eq:asymptotics-for-eigenvalues}
\end{gather}
Clearly, the latter implies that (C1) and (C2) hold irrespective
of the value of $l$.
In order to analyze the validity of (C3), let us compute the following 
functions associated with the spectra of $H_{l,0}$ and $H_{l,\beta_l}$,
\[
h_0(z) := \prod_{n=1}^\infty\left[1-\frac{z}{(j_{l+\frac12,n})^2}\right],
\quad
h_{\beta_l}(z)
	:= z\prod_{n=1}^\infty\left[1-\frac{z}{(j_{l+\frac32,n})^2}\right].
\]
Using \cite[9.5.10]{abramowitz}, we obtain
\[
h_0(z) = 2^{l+\frac12}\Gamma(l+\tfrac32)
			z^{-\frac{2l+1}{4}}J_{l+\frac12}(\sqrt{z}),
\,\,
h_{\beta_l}(z) = 2^{l+\frac32}\Gamma(l+\tfrac52)
					z^{-\frac{2l-1}{4}}J_{l+\frac32}(\sqrt{z}).
\]
We also need the derivative of $h_0(z)$. Using
\cite[9.1.30]{abramowitz},
\[
\frac{d}{dz}\left[z^{-\frac{2l+1}{4}}J_{l+\frac12}(\sqrt{z})\right]
	= -\frac12 z^{-\frac{2l+3}{4}}J_{l+\frac32}(\sqrt{z}).
\]
Therefore,
\[
h'_0(z) = - 2^{l-\frac12}\Gamma(l+\tfrac32)
			z^{-\frac{2l+3}{4}}J_{l+\frac32}(\sqrt{z}).
\]
We want to analyze the convergence of the series,
\[
\sum_{x_m\in\spec(H_{l,\beta})}\frac{1}{x_m^{2n}\abs{h_{\beta_l}(x_m)h'_0(x_m)}},
\quad n\in\Z^+.
\]
Since
\[
x_m^{2n}\abs{h_{\beta_l}(x_m)h'_0(x_m)}
	= 2^{2l+1}\Gamma(l+\tfrac32)\Gamma(l+\tfrac52)x_m^{2n-l-\frac12}
		\big[J_{l+\frac32}(\sqrt{x_m})\big]^2,
\]
this problem reduces to the convergence of
\begin{equation}
\label{eq:convergence-c3}
\sum_{x_m\in\spec(H_{l,\beta})}\frac{1}{x_m^{2n-l-\frac12}
		\big[J_{l+\frac32}(\sqrt{x_m})\big]^2}.
\end{equation}
Due to \eqref{eq:asymptotics-for-bessel} and \eqref{eq:asymptotics-for-eigenvalues}, 
it follows on one hand that
\[
\big[J_{l+\frac32}(\sqrt{x_m})\big]^2
	= \frac{4}{\pi^2(2m+l)}\underbrace{\cos^2(\pi(m-1))}_{\ds{\equiv 1}}
		+ O(m^{-2}).
\]
On the other hand,
\[
x_m^{2n-l-\frac12} = (j_{l+\frac12,m})^{4n-2l-1}
	= \big[\tfrac{\pi}{4}(2m+l)\big]^{4n-2l-1}\left[1 + O(m^{-2})\right].
\]
Thus the convergence of \eqref{eq:convergence-c3} becomes the convergence of
\[
\sum_{m=M_0}^\infty\frac{1}{(2m+l)^{4n-2l-2}\left[1 + O(m^{-1})\right]},
\]
for some $M_0$ large enough, which in turn yields our assertion.
\end{proof}

\begin{remark}
\begin{enumerate}
\item It is clear that no symmetric operator $H_l$ associated with
  $\tau_l$ can be $0$-entire (or entire according to Krein's
  definition).
\item As is readily seen from
  (\ref{eq:differential-expression-free}) and
  (\ref{eq:boundary-condition-at-0}), the case with $l=0$ corresponds
  to (minus) the Laplacian operator in $[0,1]$ with Dirichlet boundary
  condition at the left endpoint.  The operator acting the same but
  with Neumann boundary condition is discussed in \cite{IV}. As in the
  present work, there we obtain a $1$-entire operator.
\end{enumerate}
\end{remark}

\section{The radial Schr\"odinger operator: Adding a perturbation}

In this section we prove the main result of this work, namely, a
generalization of Theorem~\ref{thm:free-operator-is-n-entire}
to include a potential energy function. In other
words, we now discuss selfadjoint operators with separated
boundary conditions that arise from the differential expression
\begin{equation}
\label{eq:differential-expression}
\tau := \tau_l + q(x)
	  = -\frac{d^2}{dx^2} + \frac{l(l+1)}{x^2} + q(x).
\quad x\in(0,1),\quad  l\ge-\frac12.
\end{equation}
To start with, we assume that $q(x)$ is a real function such that
$\tilde{q}(x)\in L_1(0,1)$, where
\begin{equation}
\label{eq:definition-q-tilde}
\tilde{q}(x):=\left\{\begin{array}{ll}
					 xq(x) & l>-\tfrac12,
					 \\[1mm]
					 x(1-\log x)q(x) & l=-\tfrac12.
			         \end{array}
			   \right.
\end{equation}
Under this hypothesis, it is shown in
\cite[Theorem 2.4]{kostenko} that $\tau$ is regular at $x=1$, and in 
the limit point case (resp. limit circle case) at $x=0$ for
$l\ge 1/2$ (resp. $l\in[-1/2,1/2)$). As in the preceding
section, we impose the boundary condition \eqref{eq:boundary-condition-at-0}
in the latter case. In this way, $\tau$ gives rise to a family of
selfadjoint operators $H_\beta$, for $\beta\in[0,\pi)$, associated to the
boundary conditions $\varphi(1)\cos\beta = \varphi'(1)\sin\beta$.
These operators are the
canonical selfadjoint extensions of a certain closed, regular, symmetric
operator $H$, having deficiency indices both equal to $1$.

According to \cite[Lemma 2.2]{kostenko}, the equation
$\tau\varphi(x)=z\varphi(x)$ has a solution $\xi(z,x)$,
entire with respect to $z$, that lies in $L_2(0,1)$. This function
satisfies the following estimates \cite[Equations 2.18 and 2.10]{kostenko},
\begin{gather}
\abs{\xi(z,x)-\xi_l(z,x)}\le
	C\left(\frac{x}{1+\sqrt{\abs{z}}x}\right)^{l+1}
		e^{\abs{\im(\sqrt{z})}x}
		\int_0^x\frac{\abs{\tilde{q}(y)}}{1+\sqrt{\abs{z}}y}dy,
		\label{eq:bound-1}
\\[4mm]
\abs{\xi'(z,x)-\xi'_l(z,x)}\le
	C\left(\frac{x}{1+\sqrt{\abs{z}}x}\right)^{l}
		e^{\abs{\im(\sqrt{z})}x}
		\int_0^x\frac{\abs{\tilde{q}(y)}}{1+\sqrt{\abs{z}}y}dy,
		\label{eq:bound-2}
\end{gather}
where $\xi_l(z,x)$ is given by \eqref{eq:fundamental-solution-free}.
Moreover, $\xi(z,x)$ lies in $\Ker(H^*-zI$) for all $z\in\C$. It
also yields the spectrum of $H_\beta$ since
\[
\spec(H_\beta)
	= \{\text{zeros of }\xi(z,1)\cos\beta-\xi'(z,1)\sin\beta\}.
\]


\begin{lemma}
\label{lem:bound-due-to-q}
Suppose $\tilde{q}(x)\in L_s(0,1)$, with $1\le s\le\infty$. Then, for
$r$ such that $s^{-1} + r^{-1} = 1$,
\[
\norm{\xi(w^2,\cdot)-\xi_l(w^2,\cdot)}_2
	= e^{\abs{\im w}}
        \begin{cases}
          O\left(\abs{w}^{-l-1-1/r}\right),&  s<\infty,
          \\[2mm]
          O\left(\abs{w}^{-l-2}\log\abs{w}\right),& s=\infty,
        \end{cases}
\]
for $w\in\C$, and $\abs{w}\to\infty$.
\end{lemma}
\begin{proof}
Consider the case $1<s<\infty$ since the remaining cases can be treated
analogously.

Let $v>0$. The H\"older inequality yields
\[
\int_0^1\frac{\abs{\tilde{q}(x)}}{1+vx}dx
	\le \norm{\tilde{q}}_s
		\left[\int_0^1\frac{dx}{(1+vx)^r}\right]^{1/r}.
\]
Integrating the second factor we readily obtain the estimate
\begin{equation}
\label{eq:behaviour-q}
\int_0^1\frac{\abs{\tilde{q}(x)}}{1+vx}dx
		= O\left(v^{-1/r}\right),\quad v\to\infty.
\end{equation}
Now, in view of \eqref{eq:bound-1} and the latter estimate, we have
\begin{multline*}
\norm{\xi(w^2,\cdot)-\xi_l(w^2,\cdot)}^2_2
\\
\begin{aligned}
&\le
	C^2e^{2\abs{\im w}}\int_0^1\underbrace{\left(
	\frac{x}{1+\abs{w}x}
	\right)^{2(l+1)}}_{\ds{=:\abs{a(x)}}}
	\underbrace{\left(
	\int_0^x\frac{\abs{\tilde{q}(y)}}
	{1+\abs{w}y}dy\right)^2}_{\ds{=:\abs{b(x)}}}dx
\\
&\le
	C^2e^{2\abs{\im w}}\norm{a}_1\norm{b}_\infty
\\
&\le
	C^2e^{2\abs{\im w}}
	\underbrace{\int_0^1\left(\frac{x}{1+\abs{w}x}
	\right)^{2(l+1)}dx}_{\ds{=O\left(\abs{w}^{-2(l+1)}\right)}}
	\underbrace{\left(
	\int_0^1\frac{\abs{\tilde{q}(y)}}{1+\abs{w}y}dy
	\right)^2}_{\ds{=O\left(\abs{w}^{-2/r}\right)}}.
    \end{aligned}
  \end{multline*}
This completes the proof.
\end{proof}

Let us consider the following transform of a function $\varphi(x)\in L^2(0,1)$,
\begin{equation}
\label{eq:map}
\varphi(x)\mapsto\hat{\varphi}(z):=\int_0^1\xi(z,x)\varphi(x)dx.
\end{equation}
According to \cite[Theorem 3.2]{eckhardt}, the linear set
\[
\cB := \left\{\hat{\varphi}(z): \varphi(x)\in L^2(0,1)\right\}
\]
coincides with the de Branges space $\cB(e)$ generated by the
Hermite-Biehler function
\[
e(z) := \xi(z,1) + \I\xi'(z,1),\quad z\in\C.
\]
Since moreover $\xi(z,\cdot)\in\ker(H^*-zI)$, it is immediate to see that 
this construction is a particular instance of the abstract functional model 
discussed in \cite{IV}.

The proof of our next theorem makes use of an analogue
of the Paley-Wiener Theorem for the Hankel transform of a function
with compact support, a result due to Griffith \cite{griffith}.
For the sake of convenience we recall it (taken from
\cite{zemanian} with minor changes).

\begin{thm}[Griffith]
Let $z=x+\I y$ and assume $l>-1$. A function $f(z)$ has the
representation
\[
f(z) = \int_0^b\sqrt{zx}\,J_{l+\frac12}(zx)\varphi(x)dx
\]
with $b>0$ and $\varphi(x)\in L^2(0,b)$ if and only if
$f(x)\in L^2(0,\infty)$, $z^{-l-1}f(z)$ is an even entire function,
and there exists a constant $C>0$ such that
$\abs{f(z)}\le Ce^{b\abs{y}}$ for all $z\in\C$.
\end{thm}

We also recall the following estimate, taken from
\cite[Lemma A.1]{kostenko},
\begin{equation}
\label{eq:bound-0}
\abs{\xi_l(z,x)}
	\le C\left(\frac{x}{1+\sqrt{\abs{z}}x}\right)^{l+1}
		e^{\abs{\im(\sqrt{z})}x}.
\end{equation}
This bound holds for $l\ge -1/2$.

\begin{theorem}
\label{thm:spaces-equality}
Assume $\tilde{q}(x)\in L_s(0,1)$, with $2<s\le\infty$. Then,
$\cB=\cB_0$ (as sets), where $\cB_0$ is given by
\[
\cB_0 = \left\{\hat{\varphi}(z)=\int_0^1\xi_l(z,x)\varphi(x)dx :
		\varphi(x)\in L^2(0,1)\right\},
\]
with $\xi_l(z,x)$ given by \eqref{eq:fundamental-solution-free}.
\end{theorem}
\begin{proof}
To simplify the exposition, we restrict the proof to $2<s<\infty$. The
argumentation for $s=\infty$ is essentially the same.
We follow the arguments of the proof of
\cite[Theorem 4.1]{remling}.

We prove first that $\cB\subset\cB_0$. Consider
$f(z)\in\cB$ so
\[
f(z) = \int_0^1\xi(z,x)\varphi(x)dx, \quad \varphi(x)\in L^2(0,1).
\]
Define $g(w):= f(w^2)$. This function is clearly entire and even.
Moreover, 
\begin{align*}
\abs{g(w)}
	&\le \int_0^1 \abs{\xi(w^2,x)\varphi(x)} dx
	\\[2mm]
	&\le \left(\norm{\xi(w^2,\cdot)-\xi_l(w^2,\cdot)}_2 +
			\norm{\xi_l(w^2,\cdot)}_2\right)\norm{\varphi}_2
\end{align*}
By applying Lemma~\ref{lem:bound-due-to-q} to the first term of the
estimate, and using (\ref{eq:bound-0}) to obtain an upper bound for
the second term, one obtains
\begin{equation*}
  \abs{g(w)}\le C\norm{\varphi}_2e^{\abs{w}},
\end{equation*}
thus $g(w)$ is of exponential type not greater than $1$. Also,
\eqref{eq:bound-1} along with Lemma~\ref{lem:bound-due-to-q}
yields
\begin{align*}
g(w)
	&=  \int_0^1\xi_l(w^2,x)\varphi(x)dx +
		\int_0^1\left[\xi(w^2,x)-\xi_l(w^2,x)\right]\varphi(x)dx
	\\
	&= \sqrt{\frac{\pi}{2}}w^{-l-1}
		\int_0^1\sqrt{wx}\,J_{l+\frac12}(wx)\varphi(x)dx +
		O\left(\frac{1}{\abs{w}^{l+1+1/r}}\right),
\end{align*}
for $w\in\R$, $w\to\infty$, and $r<2$. This implies, by virtue
of Griffith's Theorem, that the function $h(w)=w^{l+1}g(w)$ lies
in $L^2(0,\infty)$. Since clearly $w^{-l-1}h(w)$ is even and entire,
again by Griffith's Theorem there exists $\eta(x)\in L^2(0,1)$ such
that
\[
h(w) = \sqrt{\frac{\pi}{2}}
		\int_0^1\sqrt{wx}\,J_{l+\frac12}(wx)\eta(x)dx.
\]
Therefore,
\[
f(z) = \int_0^1\xi_l(z,x)\eta(x)dx,
\]
as desired.

We now turn to the converse inclusion, that is, $\cB_0\subseteq\cB$. 
To show this, we check the conditions stated in Remark~\ref{rem:alternative-dB}.
First, we verify that
\eqref{eq:dB-integrability} holds
for all $f(z)\in\cB_0$. For this it will suffice to prove that
\begin{equation}
\label{eq:comparison-hb}
\underset{x\to\pm\infty}{\lim\inf}\abs{\frac{e(x)}{e_0(x)}}>0.
\end{equation}
Here,
\[
e(z):= \xi(z,1) + \I\xi'(z,1),\qquad e_0(z):= \xi_l(z,1) + \I\xi'_l(z,1),
\]
are the Hermite-Biehler functions associated to the de Branges spaces
$\cB$ and $\cB_0$ respectively. Note that, because of \eqref{eq:bound-1},
\eqref{eq:bound-2} and \eqref{eq:behaviour-q}, one has
\begin{equation}
\label{eq:quotient-hb-functions}
\abs{\frac{e(x)}{e_0(x)}}^2 =
	\frac{\abs{\xi_l(x,1) + O\left(\abs{x}^{-\frac{l+1+1/r}{2}}\right)}^2
	+ \abs{\xi_l'(x,1) + O\left(\abs{x}^{-\frac{l+1/r}{2}}\right)}^2}
	{\abs{\xi_l(x,1)}^2 + \abs{\xi_l'(x,1)}^2},
\end{equation}
for $x\in\R$. We will analyze the cases $x\to-\infty$ and $x\to\infty$
separately.

Set $x=-w^2$ with $w\in\R^+$. Recalling
\eqref{eq:fundamental-solution-entire},
\begin{align*}
\xi_l(-w^2,1) &= \frac{\sqrt{\pi}}{2^{l+1}}g_l(-w^2,1),\\
\xi_l'(-w^2,1) &= (l+1)\frac{\sqrt{\pi}}{2^{l+1}}g_l(-w^2,1)
	+ \frac{\sqrt{\pi}}{2^{l+1}}g_l'(-w^2,1),
\end{align*}
where $g_l(z,x)$ is given by \eqref{eq:fundamental-solution-entire} so
it readily follows that $\xi_l(-w^2,1)\to\infty$ and
$\xi_l'(-w^2,1)\to\infty$ as $w\to\infty$. Hence we obtain
\begin{equation}
\label{eq:im-so-bored}
\lim_{w\to\infty}\abs{\frac{e(-w^2)}{e_0(-w^2)}}^2 = 1.
\end{equation}

Now consider $x= w^2$ with $w\in\R^+$. By \cite[9.2.1]{abramowitz},
\[
J_m(w) = \sqrt{\frac{2}{\pi w}}
		 \left[\cos\left(w-\tfrac12m\pi-\tfrac14\pi\right)
		 + O\left(w^{-1}\right)\right],\quad w\to\infty,
\]
thus it follows that
\begin{gather*}
\xi_l(w^2,1)  = w^{-l-1}\sin\left(w-\tfrac12\pi l\right)
				+ O\left(w^{-l-2}\right),\quad \text{and}
\\[2mm]
\xi_l'(w^2,1) = w^{-l}\cos\left(w-\tfrac12\pi l\right)
				+ O\left(w^{-l-1}\right),
\end{gather*}
as $w\to\infty$. Inserting these expressions into
\eqref{eq:quotient-hb-functions} one gets
\begin{equation}
\label{eq:quotient-again}
\abs{\frac{e(w^2)}{e_0(w^2)}}^2\!\! =
	\frac{\abs{\sin\left(w\!-\!\frac{\pi l}{2}\right)
	\!+\!O\left(w^{-\frac{1}{r}}\right)}^2
	\!\!+ \abs{w\cos\left(w\!-\!\frac{\pi l}{2}\right)
	\!+\!O\left(w^{1-\frac{1}{r}}\right)}^2}
	{\abs{\sin\left(w\!-\!\frac{\pi l}{2}\right)
	\!+ O\left(w^{-1}\right)}^2
	\!\!+ \abs{w\cos\left(w\!-\!\frac{\pi l}{2}\right)
	\!+ O\left(1\right)}^2}.
\end{equation}
Suppose there exists a positive unbounded sequence
$\{w_n\}_{n=1}^\infty$ such that \eqref{eq:quotient-again} goes to
zero as $n\to\infty$. Necessarily, $w_n\cos\left(w_n-\tfrac12\pi
  l\right)$ is unbounded since otherwise it would follow that
$\cos\left(w_n-\tfrac12\pi l\right)\to 0$ as $n\to\infty$, along with
$\sin\left(w_n-\tfrac12\pi l\right)\to 1$, implying that
\eqref{eq:quotient-again} is bounded below away from zero. But if
$w_n\cos\left(w_n-\tfrac12\pi l\right)$ is unbounded, then
\eqref{eq:quotient-again} goes to 1, again contradicting our
assumption. Thus, we obtain
\[
\underset{w\to\infty}{\lim\inf}\abs{\frac{e(w^2)}{e_0(w^2)}}^2 >0,
\]
which, together with \eqref{eq:im-so-bored}, gives
\eqref{eq:comparison-hb}. Thus, we have established
\eqref{eq:dB-integrability}.

To complete the proof, we now check \eqref{eq:dB-Cauchy}. To this end,
we calculate
\begin{equation*}
  \int_{\Gamma_A}\frac{f(s)}{e(s)(s-z)}ds,\quad z\in\C^+,\quad f(z)\in\cB_0,
\end{equation*}
where $\Gamma_A=[-A,A]\cup\{Ae^{\I\phi}: \phi\in[0,\pi]\}$, and
attempt to let $A\to\infty$. For the function $f^\#(z)$ a similar
argument can be carried out.  By a standard argumentation, to prove
\eqref{eq:dB-Cauchy}, it suffices to show that
\begin{equation}
 \label{eq:limit-integral-weak}
 \lim_{n\to\infty}\int_0^\pi\abs{\frac{f(A_ne^{i\phi})}{e(A_ne^{i\phi})}}d\phi=0
\end{equation}
for some sequence $A_n\to\infty$.

Now, we implement the change of variable $z=R^2e^{\I 2\theta}$,
$0\le\theta\le\pi/2$. According to \cite[Chapter 4, Section
9]{olver}, it holds true that
\begin{equation}
  \label{eq:olver-asymptotics}
  J_m(Re^{\I \theta}) = \sqrt{\frac{2}{\pi R}}e^{-\I \theta/2}
		 \left[\cos\left(Re^{\I\theta}-\tfrac12m\pi-\tfrac14\pi\right)
		 + e^{R\sin\theta}O_\epsilon\left(R^{-1}\right)\right],
\end{equation}
as $R\to\infty$, provided that
\begin{equation}
\label{eq:excluded-sector}
 0\le\theta\le\pi/2-\epsilon\,,\qquad\epsilon>0\,.
\end{equation}
Note that we have stressed the dependence on $\epsilon$ of the
constant implicit in the asymptotic remainder by means of a
subscript.

Next, due to \eqref{eq:bound-1}, \eqref{eq:bound-2}, \eqref{eq:behaviour-q} and
\eqref{eq:olver-asymptotics}, the following holds in the
sector (\ref{eq:excluded-sector}):
\begin{align*}
  e(R^2e^{\I 2\theta})
  =&\, \xi_l(R^2e^{\I 2\theta},1)
  	 + \I\xi_l'(R^2e^{\I 2\theta},1)
  	 +  e^{R\sin\theta}O\left(\frac{1}{R^{l+1/r}}\right)
  	 \\[2mm]
  =&\, R^{-(l+1)}e^{\I(l+1)\theta}\left[\sin\left(Re^{\I\theta}-
  \frac{l\pi}{2}\right)+ e^{R\sin\theta}O_\epsilon
  \left(\frac{1}{R}\right)\right]
  	 \\
   &+ \I(l+1)R^{-(l+1)}e^{-\I\theta(l+1)}
   		\left[\sin\left(Re^{\I\theta}-\frac{l\pi}{2}\right)
   	+  e^{R\sin\theta}O_\epsilon\left(\frac{1}{R}\right)\right]
  	 \\
   &+ \I R^{-l}e^{-\I l\theta}
   		\left[\cos\left(Re^{\I\theta}-\frac{l\pi}{2}\right)
   	+
        e^{R\sin\theta}O_\epsilon\left(\frac{1}{R}\right)\right]
\\
&+ e^{R\sin\theta}O\left(\frac{1}{R^{l+1/r}}\right)
\end{align*}
as $R\to\infty$. Notice that $1<r<2$ so the latter error term is the
dominant one. Thus, asymptotically,
\begin{equation}
  \label{eq:bound-from-below}
  \abs{e(R^2e^{\I 2\theta})}\ge R^{-l}\abs{\abs{\cos\left(Re^{\I\theta}-
  l\pi/2\right)} - e^{R\sin\theta}O_\epsilon
  \left(\frac{1}{R^{l+1/r}}\right)}.
\end{equation}

We want to obtain lower bounds for $\abs{e(R^2e^{\I 2\theta})}$, with the
help of (\ref{eq:bound-from-below}), separately for large angles that satisfy
\begin{equation}
  \label{eq:large-angles}
   e^{-2R\sin\theta}< 1-\delta
\end{equation}
but
within (\ref{eq:excluded-sector}), and small ones
\begin{equation}
  \label{eq:small-angles}
   e^{-2R\sin\theta}> 1-\delta
\end{equation}
for some small $\delta>0$.

Let us first consider (\ref{eq:large-angles}). Inequality
(\ref{eq:bound-from-below}) can be rewritten as
\begin{align*}
  \abs{e(R^2e^{\I 2\theta})}\ge\; &
  \frac{1}{2}e^{R\sin\theta}R^{-l}\\
&\times
  \abs{\abs{1+e^{\I
        2(R\cos\theta-l\pi/2)}e^{-2R\sin\theta}}-
O_\epsilon\left(\frac{1}{R^{l+1/r}}\right)}.
\end{align*}
Then there exists $R_\epsilon$ such that, for all $R> R_\epsilon$,
\begin{equation*}
  \abs{e(R^2e^{\I 2\theta})}\ge
  \frac{1}{4}e^{R\sin\theta}R^{-l}
  \abs{1+e^{\I
        2(R\cos\theta-l\pi/2)}e^{-2R\sin\theta}}.
\end{equation*}
This implies
\begin{equation*}
  \abs{e(R^2e^{\I 2\theta})}
  \ge \frac{1}{4}e^{R\sin\theta}R^{-l}
  \left(1-e^{-2R\sin\theta}\right)
\end{equation*}
so that, under (\ref{eq:large-angles}), one obtains
\begin{equation}
  \label{eq:large-estimate}
  \abs{e(R^2e^{\I 2\theta})}\ge\frac{1}{4}\delta R^{-l}e^{R\sin\theta}
\end{equation}
for $R>R_\epsilon$.

Let us now consider small angles, that is (\ref{eq:small-angles}), keeping
$R\ge R_\epsilon$.
Taking into account that
\begin{align*}
  \cos(Re^{\I\theta}-l\pi/2)&=
  \frac12 e^{R\sin\theta}\\ &\times
\left[\cos(R\cos\theta-l\pi/2)+e^{\I(R\cos\theta-l\pi/2)}(e^{-2R\sin\theta}-1)
\right],
\end{align*}
one deduces from
(\ref{eq:bound-from-below}) the following asymptotic estimate
\begin{equation*}
  \abs{e(R^2e^{\I 2\theta})}\ge\frac{1}{2}
  R^{-l}e^{R\sin\theta}\abs{\cos(R\cos\theta-l\pi/2)} + O_\epsilon
  \left(\frac{1}{R^l}\right).
\end{equation*}
Consider the sequence $\{R_n\}_{n=1}^\infty$ given by
$R_n:=\pi(n+l/2)$. Noticing that
\begin{align*}
  \cos(R_n\cos\theta-l\pi/2)&=\cos(n\pi\cos\theta)
  +(\cos\theta-1)\l\pi/2)\\
  &=\cos(n\pi+
\underbrace{O(n\sin^2\theta)}_{\mathclap{\ds{O(n^{-1})\text{ due to
      \eqref{eq:small-angles},}}}})
\end{align*}
we arrive at
\begin{equation}
  \label{eq:sequence-bound}
  \abs{e(R_n^2e^{\I 2\theta})}\ge
  \frac{1}{4}\delta R_n^{-l}e^{R_n\sin\theta}
\end{equation}
for $n$ sufficiently large.

On the other hand, for $f(z)\in\cB_0$,
\begin{align}
\abs{f(R^2e^{\I 2\theta})}
	&\le \int_0^1\abs{\varphi(x)}\abs{\xi_l(R^2e^{\I 2\theta},x)}dx\nonumber
	\\
	&\le C\norm{\varphi}_2\left(\frac{1}{1+R}\right)^{l+1}
		e^{R\sin\theta},\label{eq:no-cat}
\end{align}
where $\varphi(x)$ is such that $f(z)=\hat{\varphi}(z)$ and we have 
used (\ref{eq:bound-0}).
Taking into account (\ref{eq:large-estimate}) and
(\ref{eq:sequence-bound}), the last estimate implies, for some
constant $C'>0$ and all $\theta\in[0,\pi/2-\epsilon)$,
\begin{equation}
\label{eq:bounded-quotient}
\abs{\frac{f(R^2e^{\I 2\theta})}{e(R^2e^{\I 2\theta})}}
	\le C'\frac{1}{R},\quad
\abs{\frac{f(R_n^2e^{\I 2\theta})}{e(R_n^2e^{\I 2\theta})}}
	\le C'\frac{1}{R_n},
\end{equation}
where $R>R_\epsilon$ and $R_n>R_\epsilon$, for large and
small angles, respectively.

We claim that the first estimate in (\ref{eq:bounded-quotient}) holds
also for $\theta\in[\pi/2-\epsilon,\pi/2]$. For
(\ref{eq:bounded-quotient}) implies that
\begin{equation}
\label{eq:phragmen}
\abs{\sqrt{z}\frac{f(z)}{e(z)}}
\end{equation}
is bounded along the ray $\theta=\pi/2-\epsilon$ and, since
(\ref{eq:dB-integrability}) has been verified, one concludes that
(\ref{eq:phragmen}) is also bounded along the ray $\theta=\pi/2$.  Then,
the Phragm{\'e}n-Lindel{\"o}f principle \cite[Theorem 21]{levin}
implies that \eqref{eq:phragmen} is bounded inside the angle. In this
argument, a suitable branch of $\sqrt{z}$ is assumed.

Finally, using that (\ref{eq:bounded-quotient}) holds for
$\theta\in[0,\pi/2]$, one obtains (\ref{eq:limit-integral-weak}) for
$A_n=R^2_n$. Clearly, the same argumentation used above leads to the
same result for $f^\#(z)/e(z)$ since, in view of (\ref{eq:map}),
inequality (\ref{eq:no-cat}) holds when one substitutes $f(z)$ by
$f^\#(z)$ and $\varphi(x)$ by $\cc{\varphi(x)}$. Thus, we have shown
(\ref{eq:dB-Cauchy}) to be true.
\end{proof}

\begin{theorem}
\label{thm:main}
Let $l\ge -\frac12$.  Let $H$ be the regular, symmetric operator with
deficiency indices both equal to  $1$ that is associated with the formal
differential expression \eqref{eq:differential-expression}, and the
boundary condition \eqref{eq:boundary-condition-at-0} whenever
$[-1/2,1/2)$. Assume that $\tilde{q}(x)$, given by
\eqref{eq:definition-q-tilde}, belongs to $L_p(0,1)$ for some
$p>2$. Then the operator $H$ is $n$-entire if and only if
$n>\frac{l}{2}+\frac34$.
\end{theorem}
\begin{proof}
  Due to Theorem~\ref{thm:spaces-equality} the de Branges spaces
  associated with the operators $H_l$ and $H$ are the same (note that
  $\cB_0$ is the de Branges space associated with $H_l$). Then the
  corresponding spaces of $n$-associated functions also coincide. The
  assertion thus follows from
  Theorem~\ref{thm:free-operator-is-n-entire} and Theorem
  \ref{thm:n-entire}.
\end{proof}
Combining this theorem with Theorem~\ref{thm:n-entire}, it follows that
the spectra of the selfadjoint realizations $H_\beta$ have certain asymptotic
properties: 

\begin{corollary}
  \label{cor:main}
  Assume that $l\ge -\frac12$ and $\tilde{q}(x)$ lies in $L_p(0,1)$,
  with $p>2$. Then the spectra of two canonical selfadjoint extensions
  $H_{\beta_1}$, $H_{\beta_2}$ of $H$ satisfy the conditions (C1),
  (C2), (C3), of Theorem~\ref{thm:n-entire} with
  $n>\frac{l}{2}+\frac34$.
\end{corollary}

\subsection*{Acknowledgments}
J.H.T. gratefully acknowledges the kind hospitality of IIMAS--UNAM
during a visit in 2013 where most of this article was written.
The authors thank the referees for their valuable suggestions and hints to the 
literature.

\end{document}